\documentclass[journal,a4paper]{IEEEtran}
\usepackage{graphicx}
\usepackage{amsmath}
\usepackage{amssymb}
\usepackage{amsthm}
\usepackage{lineno}
\usepackage{verbatim}
\usepackage{amsfonts}
\usepackage{mathrsfs}
\usepackage{color}
\usepackage{bm}

\newtheorem{mythm}{Theorem}

\newtheorem{mydef}{Definition}

\newtheorem{myrem}{Remark}

\allowdisplaybreaks

\begin{document}

\title{Synchronization and Control for Multi-Weighted and Directed Complex Networks}

\author{Xiwei Liu,~\IEEEmembership{Senior Member, IEEE}
\thanks{This work was supported by the National Science Foundation of China under Grant No. 62073243, 61673298, Shanghai Rising-Star Program under Grant No. 17QA1404500, Natural Science Foundation of Shanghai under Grant No. 17ZR1445700, the Fundamental Research Funds for the Central Universities.}

\thanks{Xiwei Liu is with Department of Computer Science and Technology, Tongji University, and with the Key Laboratory of Embedded System and Service Computing, Ministry of Education, Shanghai 201804, China. E-mail: xwliu@tongji.edu.cn; xwliu.sh@gmail.com}
}

\maketitle
\begin{abstract}
The study of complex networks with multi-weights has been a hot topic recently. For a network with a single weight, previous studies have shown that they can promote synchronization. But for complex networks with multi-weights, there are no rigorous analysis to show that synchronization can be reached faster. In this paper, the complex network is allowed to be directed, which will make the synchronization analysis difficult for multiple couplings. In virtue of the normalized left eigenvectors (NLEVec) corresponding to the zero eigenvalue of coupling matrices, we prove that if the Chebyshev distance between NLEVec is less than some value, which is defined as the allowable deviation bound, then the synchronization and control will be realized with sufficiently large coupling strengths, i.e., all coupling matrices do accelerate synchronization. Moreover, adaptive rules are also designed for the coupling strength.
%

\end{abstract}

\section{Introduction}
Synchronization, as an interesting collective behavior of complex networks, has attracted many researchers' attention. Previous works mainly focus on the synchronization and its control for networks with a single coupling matrix, which has been widely and deeply investigated. Among them, pioneering works \cite{murray2004}, \cite{wu2005}, \cite{luchen2006} all use the normalized left eigenvectors (NLEVec) corresponding to the zero eigenvalue of coupling matrix to set up the framework of synchronization analysis. Using this NLEVec, a dummy node is defined, and the synchronization between nodes is transformed to the synchronization between nodes and this dummy node. Moreover, this NLEVec can also be used in the pinning control problem for networks, see \cite{chenliulu2007}. The main stream of investigation for synchronization and control follows these works' technical route, which makes NLEVec vital in analysis \cite{liu2010}.

However, multiple coupling is more physical than single coupling in actual networks. For example, in the social network, people interact with others by using many different ways, such as Facebook, WeChat, mail, E-mail, telephone, etc. In the past year, because of multiple sources of infection for COVID-19, like human-to-human transmission, goods-to-human transmission, environment-to-human transmission, and so on, its control and prevention is still a challenging question for our world. `Complexity arises from the diversity of the interactions between components', said in \cite{BA2021}.

Recently, the study of complex networks with multi-weights (CNMWs) has been a hot topic. \cite{lilixiang2008} and \cite{an2014} took the public traffic network as a CNMWs, where each single weight represents a transportation type, like railway, highway, airplane, and so on, then investigated its synchronization. \cite{du2019} formed a two-layer-coupled public bus and subway traffic CNMWs, and studied its synchronization control problem. \cite{zhao2015} investigated the synchronization for a network with multiple time delays. More study of synchronization for CNMWs including: fractional-order CNMWs in \cite{s2019}, finite-time synchronization in \cite{huang2018} and \cite{wang2020}, pinning control in \cite{feng2019} and \cite{wang2019}, event-triggered control in \cite{huang2020} and \cite{chenwz2020}, etc.

However, there still exists a key problem being not solved: for CNMWs without control, if the coupling matrices are asymmetric, how to prove its synchronization? In this case, different coupling matrices will cause different NLEVec, which will make the Lyapunov function hard to design. Therefore,

1. We firstly loosen the requirement of NLEVec for complex networks with a single weight (CNSW), and prove that if the Chebyshev distance between NLEVec and a normalized positive vector is less than a bound, then this vector can also be used for synchronization analysis. This fact will greatly improve the rigid property of original synchronization technique and successfully bridge the gap between single weighted and multi-weighted networks.

2. According to the above analysis, for directed CNMWs, we design a combination of different NLEVec for coupling matrices, and prove its validity for synchronization.

3. We also consider the synchronization of directed CNMWs under pinning control, and the corresponding adaptive rules for coupling strength are also investigated.

The rest is organized as follows. In Section \ref{generalization}, the generalization from NLEVec to any vector for CNSW is presented, two bounds $\mathcal{ADSB}$ and $\mathcal{ADCB}$ are defined. In Section \ref{main}, the (adaptive) synchronization and control for CNMWs is investigated, and we mainly focus on networks with two coupling matrices. Finally, some conclusions and discussions are given in Section \ref{con}.

\section{Synchronization for single-weighted and directed network}\label{generalization}

In this section, we will solve the following problem: \emph{except the NLEVec, can we find other vectors to prove the synchronization for single-weighted and directed network?}

\subsection{Synchronization problem}

Suppose the network model is described as
\begin{align}\label{single}
\dot{z}_i(t)=h(z_i(t))+c\sum_{j=1}^NG_{ij}\Gamma z_j(t), \quad i=1,\cdots,N
\end{align}
where $z_i(t)\in R^n$ is the state of node $i$, whose self behavior is defined by function $h(\cdot): R^n\to R^n$ with condition
\begin{align*}
(x-y)^T(h(x)-h(y))\le L_h(x-y)^T(x-y), \forall x,y\in R^n
\end{align*}
where $L_h>0$, and node $i$ is affected by its neighbours; $c>0$ is the coupling strength; outer coupling matrix $G=(G_{ij})\in R^{N\times N}$ is a Metzler matrix with zero-row-sum, the network is strongly connected and directed, so $G$ is asymmetric; and $\Gamma=\mathrm{diag}(\gamma_1,\cdots,\gamma_n)$ is the inner matrix with $\gamma_i>0$, $\forall i$.

To answer the proposed question, we choose a vector $\theta=(\theta_1,\cdots,\theta_N)^T\in R^{N}$ with $\sum_{i=1}^N\theta_i=1$ (normalized) and $\theta_i>0$ for $\forall i$. Now, we can define a dummy target as
\begin{align}\label{tar1}
\overline{z}(t)=\sum_{i=1}^N\theta_iz_i(t)
\end{align}

Now, we have the following theorem.
\begin{mythm}\label{generalize}
For network (\ref{single}), if matrix
\begin{align}\label{ivt-mat}
G_{\theta}=[(\Theta-\theta\theta^{T})G+G^{T}(\Theta-\theta\theta^{T})]/2
\end{align}
is negative definite in the transverse space ${\mathcal TS}=\{X|X\in R^{N}, X^T{\bf 1}=0\}$, where ${\bf 1}=(1,\cdots,1)^T$ and $\Theta=\mathrm{diag}(\theta)$, denote $\lambda_2(G_{\theta})$ as the largest eigenvalue (EVal) of $G_{\theta}$ in $\mathcal{TS}$ (also the second largest Eval of $G_{\theta}$ in the whole space), then exponential synchronization (Expo-Syn) can be realized if
\begin{align}\label{condition1}
L_h+c\lambda_2(G_{\theta})\min_{k}\gamma_k/\|\Theta-\theta\theta^T\|_2<0.
\end{align}
\end{mythm}

\begin{proof} 
Define
\begin{align}\label{lya}
V(t)=&\frac{1}{2}\sum_{i=1}^N\theta_i[z_i(t)-\overline{z}(t)]^T[z_i(t)-\overline{z}(t)]\nonumber\\
=&\frac{1}{2}Z(t)^T[(\Theta-\theta\theta^T)\otimes I_n]Z(t),
\end{align}
where $Z(t)=(z_1(t)^T,\cdots,z_N(t)^T)^T$. When synchronization is reached, $V(t)=0$. Therefore, suppose synchronization has not been realized, differentiating $V(t)$ along (\ref{single}),
\begin{align*}
&\dot{V}(t)=\sum_{i=1}^N\theta_i[z_i(t)-\overline{z}(t)]^T[\dot{z}_i(t)-\dot{\overline{z}}(t)]\nonumber\\
=&\sum_{i=1}^N\theta_i[z_i(t)-\overline{z}(t)]^T[h(z_i(t))+c\sum_{j=1}^NG_{ij}\Gamma z_j(t)-\dot{\overline{z}}(t)]\nonumber\\
=&\sum_{i=1}^N\theta_i[z_i(t)-\overline{z}(t)]^T[h(z_i(t))-h(\overline{z}(t))]\nonumber\\
&+c\sum_{i=1}^N\theta_i[z_i(t)-\overline{z}(t)]^T\sum_{j=1}^NG_{ij}\Gamma z_j(t)\nonumber\\
\le&L_h\sum_{i=1}^N\theta_i[z_i(t)-\overline{z}(t)]^T[z_i(t)-\overline{z}(t)]\nonumber\\
&+cZ(t)^T\bigg[[(\Theta-\theta\theta^T)G]\otimes \Gamma\bigg]Z(t)\nonumber\\
=&2L_hV(t)+cZ(t)^T[G_{\theta}\otimes \Gamma]Z(t)\nonumber\\
\le&2[L_h+c\lambda_2(G_{\theta})\min_{k}\gamma_k/\|\Theta-\theta\theta^T\|_2]V(t)<0.
\end{align*}
Therefore, exponential synchronization will be achieved.
\end{proof}

\begin{myrem}
If the vector $\theta$ is chosen as NLEVec $\xi>0$, i.e., $\xi^TG=0$, then the matrix $G_{\theta}$ will become
\begin{align}\label{ref-mat}
G_{\xi}=&[(\Xi-\xi\xi^{T})G+G^{T}(\Xi-\xi\xi^{T})]/2\nonumber\\
=&(\Xi G+G^{T}\Xi)/2.
\end{align}
According to the definition of NLEVec $\xi$, we have that $G_{\xi}$ does be negative definite in $\mathcal{TS}$. This fact has been widely used in previous works \cite{murray2004}-\cite{liu2010}.
\end{myrem}

Since for any vector $\theta$, the condition $G_{\theta}$ is negative definite in $\mathcal{TS}$ cannot be ensured, and also notice that $G_{\xi}$ does be negative definite in $\mathcal{TS}$, we can give some scopes of $\theta$ by comparing with $\xi$. Denote the error
\begin{align}\label{error}
\delta=\theta-\xi=(\delta_1,\cdots,\delta_N)^T,
\end{align}
since $\theta$ and $\xi$ are all normalized vectors, so the sum of all elements in $\delta$ is zero, and let $\Delta=\mathrm{diag}(\delta)$, then
\begin{align*}
G_{\theta}=[(\Delta-\theta\delta^{T})G+G^{T}(\Delta-\delta\theta^{T})]/2+G_{\xi}.
\end{align*}
Therefore, for any vector $y$ in $\mathcal{TS}$, we have
\begin{align}\label{process}
y^TG_{\theta}y=&y^T\bigg[[(\Delta-\theta\delta^{T})G+G^{T}(\Delta-\delta\theta^{T})]/2+G_{\xi}\bigg]y\nonumber\\
\le &y^T(\Delta-\theta\delta^{T})Gy+\lambda_{2}(G_{\xi})y^Ty\nonumber\\
\le&\|(\Delta-\theta\delta^{T})G\|_2y^Ty+\lambda_{2}(G_{\xi})y^Ty\nonumber\\
\le&\bigg[\sqrt{N}\|(\Delta-\theta\delta^{T})G\|_1+\lambda_{2}(G_{\xi})\bigg]y^Ty\nonumber\\
\le&\bigg[\sqrt{N}[\|\Delta\|_1+\|\theta\delta^{T}\|_1]\|G\|_1+\lambda_{2}(G_{\xi})\bigg]y^Ty\nonumber\\
=&\bigg[2\sqrt{N}\max_{k}|\delta_k|\cdot\|G\|_1+\lambda_2(G_{\xi})\bigg]y^Ty
\end{align}

Therefore, a sufficient condition can be stated as: when
\begin{align}\label{estimate}
\max_{i}|\delta_i|\le \frac{|\lambda_2(G_{\xi})|}{2\sqrt{N}\|G\|_1},
\end{align}
the matrix $G_{\theta}$ is negative definite in $\mathcal{TS}$. \emph{We can conclude that, for the proposed question, we have a positive answer, there do exist vectors $\theta$, if the deviation for $\theta$ from NLEVec $\xi$ is small enough (satisfying the inequality (\ref{estimate})), the matrix $G_{\theta}$ does be negative definite in $\mathcal{TS}$, which can be used for the analysis of synchronization}.

\begin{myrem}
More accurate estimation may be obtained in (\ref{process}) by using different norms and inequalities, interested readers are encouraged to improve this estimation result (\ref{estimate}).
\end{myrem}

\begin{mydef}
The allowable deviation synchronization bound ($\mathcal{ADSB}$) for matrix $G$ from its NLEVec $\xi$ is:
\begin{align}\label{bound}
\mathcal{ADSB}(G)=\frac{|\lambda_2(G_{\xi})|}{2\sqrt{N}\|G\|_1}
\end{align}
\end{mydef}
Since $N$ is the dimension, and $\xi$ is its NLEVec, therefore, $\mathcal{ADSB}(G)$ is completely determined by matrix $G$.

In the next, we will present an example to illustrate the correctness of above theory and analysis.

{\bf Example 1:} Consider the following coupling matrix
\begin{align}
G=\left[
\begin{array}{ccc}
-3& 1& 2\\
2& -4& 2\\
1& 1& -2
\end{array}
\right]
\end{align}
NLEVec of $G$ is $\xi=(0.3, 0.2,0.5)^T$; EVal of matrix $G_{\xi}$ in (\ref{ref-mat}) are: $0, -1.1768, -1.5232$, so $\lambda_2(G_{\xi})=-1.1768$.

According to the estimation (\ref{estimate}), Chebyshev distance
\begin{align*}
\max_{i}|\delta_i|\le \frac{1.1768}{2\sqrt{3}\cdot 6}=0.0566.
\end{align*}
Now, we can choose vector $\theta=(0.25, 0.25, 0.5)^T$, then EVal of matrix $G_{\theta}$ in (\ref{ivt-mat}) are: $0, -1.2096, -1.5404$. On the other hand, simple computer program does find many normalized vectors $\theta$, such that $G_{\theta}$ has positive EVal. For example, if $\theta=(0.0025, 0.52, 0.4775)^T$, EVal of $G_{\theta}$ would be: $0.004, 0, -2.2612$.

\begin{myrem}
Notice the fact that NLEVec $\xi$ and $\theta$ are both positive scalars, so the angle between them must be acute. In the above example, we can calculate the angle between NLEVec $\xi$ and the valid vector $\theta=(0.25, 0.25, 0.5)^T$ is $0.115$ rad; on the other hand, the angle between NLEVec $\xi$ and the invalid vector $\theta=(0.0025, 0.52, 0.4775)^T$ is $0.6611$ rad, so an intuitive understanding of the limitation of valid vectors $\theta$ is that they should near the NLEVec as close as possible. Of course, this does not mean that NLEVec is the optimal direction. Considering the above example, the second largest eigenvalue of $A_{\theta}$ with $\theta=(0.25, 0.25, 0.5)^T$ is $-1.2096$, while the second largest eigenvalue of $A_{\xi}$ is just $-1.1768$. Therefore, \emph{the optimal direction may not be the NLEVec, but from the viewpoint of theoretical analysis, NLEVec is an ideal reference direction for investigation}.
\end{myrem}

\subsection{Control problem}

Next, we consider the pinning control problem. The network with pinning control added on the first node is
\begin{align}\label{control}
\dot{z}_i(t)=h(z_i(t))+c\sum\limits_{j=1}^N\tilde{G}_{ij}\Gamma z_j(t)
\end{align}
where $z(t)$ is the synchronization target with $\dot{z}(t)=h(z(t))$, and the new matrix $\tilde{G}=G-\mathrm{diag}(d,0,\cdots,0)$ with $d>0$.

According to the result in \cite{chenliulu2007}, the matrix
\begin{align}\label{control-matrix}
_{\xi}\tilde{G}=\frac{\Xi\tilde{G}+\tilde{G}^T\Xi}{2}
\end{align}
is negative definite, where $\xi$ is the NLEVec for $G$. Therefore, for any vector $y\in R^N$, let $\Delta=\Theta-\Xi$, where $\Theta=\mathrm{diag}(\theta)$ and $\Xi=\mathrm{diag}(\xi)$, with the same process as (\ref{process}), we have
\begin{align}\label{process2}
y^T(_{\theta}\tilde{G})y=&y^T\frac{\Theta\tilde{G}+\tilde{G}^T\Theta}{2}y=y^T\bigg[\frac{\Delta\tilde{G}+\tilde{G}^T\Delta}{2}+(_{\xi}\tilde{G})\bigg]y\nonumber\\
\le &y^T(\Delta\tilde{G})y+\lambda_{\max}(_{\xi}\tilde{G})y^Ty\nonumber\\
\le&[\sqrt{N}\max_{i}|\delta_i|\cdot\|\tilde{G}\|_1+\lambda_{\max}(_{\xi}\tilde{G})]y^Ty,
\end{align}
where $\lambda_{\max}(\cdot)$ means the largest EVal of the matrix.

\begin{mydef}
The allowable deviation control bound ($\mathcal{ADCB}$) for matrix $\tilde{G}$ from the NLEVec $\xi$ for $G$ is:
\begin{align}\label{cbound}
\mathcal{ADCB}(\tilde{G})=\frac{|\lambda_{\max}(_{\xi}\tilde{G})|}{\sqrt{N}\|\tilde{G}\|_1}
\end{align}
\end{mydef}

\begin{mythm}
For network (\ref{control}), suppose the Chebyshev distance $\max_{i}|\theta_i-\xi_i|\le \mathcal{ADCB}(\tilde{G})$, then Expo-Syn can be realized if $L_h+c\lambda_{\max}(_{\theta}\tilde{G})\min_{k}(\gamma_k)/\max_{i}{\theta_i}<0$.
\end{mythm}

\begin{proof}
Denote $\tilde{z}_i(t)=z_i(t)-z(t)$, and let
\begin{align}\label{lya-control}
W(t)=\frac{1}{2}\sum_{i=1}^N\theta_i\tilde{z}_i(t)^T\tilde{z}_i(t).
\end{align}
Then,
\begin{align*}
\dot{W}(t)\le& 2L_hW(t)+c\sum_{i=1}^N\sum_{j=1}^N\theta_i\tilde{z}_i(t)^T\tilde{G}_{ij}\Gamma\tilde{z}_j(t)\\
\le&2[L_h+c\lambda_{\max}(_{\theta}\tilde{G})\min_{k}(\gamma_k)/\max_{i}{\theta_i}]W(t)\le 0
\end{align*}
The proof is finished.
\end{proof}

\begin{myrem}
According to the Proposition 1 in \cite{chenliulu2007}, $\tilde{G}$ is a M-matrix, so there must exist a diagonal matrix $P$, such that $P\tilde{G}+\tilde{G}^TP$ is negative definite, and we can also choose the vector $\theta$ as $\theta_i=P(i,i)$. This fact means that, for control problem, NLEVec is not the unique choice, there are also many other choices, the reason we choose it is to keep consistence with the above synchronization analysis.
\end{myrem}

\section{Synchronization and control for multi-weighted and directed network}\label{main}
With the relaxation of NLEVec $\xi$ to vector $\theta$ satisfying condition (\ref{estimate}), we can use this $\theta$ to investigate the synchronization for directed CNMWs.

For CNSW, any coupling matrix, no matter it is symmetric or asymmetric, is beneficial for synchronization, so we raise the following question: \emph{for CNMWs, which have more than one coupling matrices, would they synchronize faster? In other words, do the coupling matrices all promote the synchronization process? If this fact is true, how to prove it?}

{\bf Example 2:} Before considering this problem in theory, we consider a simple example from real simulations:
\begin{align}\label{example}
\dot{z}_i(t)=h(z_i(t))+\sum_{j=1}^3G_{ij}^{1} \Gamma^1z_j(t)+\sum_{j=1}^3G_{ij}^{2}\Gamma^2z_j(t),
\end{align}
where $z_i(t)\in R^3, i=1,2,3$, $h(\cdot)$ is the Lorenz oscillator,
\begin{align*}
&G^{1}=\left[
\begin{array}{ccc}
-3& 1& 2\\
2& -4& 2\\
1& 1& -2
\end{array}
\right],
&G^{2}=\left[
\begin{array}{ccc}
-2& 1& 1\\
1& -2& 1\\
1& 1& -2
\end{array}
\right],
\end{align*}
and $\Gamma^1=\mathrm{diag}([1,2,1])$ and $\Gamma^2=I_3$. Simulations show that CNMWs can synchronize faster than single weighted network, see Fig. 1, which means that both coupling matrices are beneficial for synchronization.

\begin{figure}
\begin{center}
\includegraphics[width=0.5\textwidth,height=0.3\textheight]{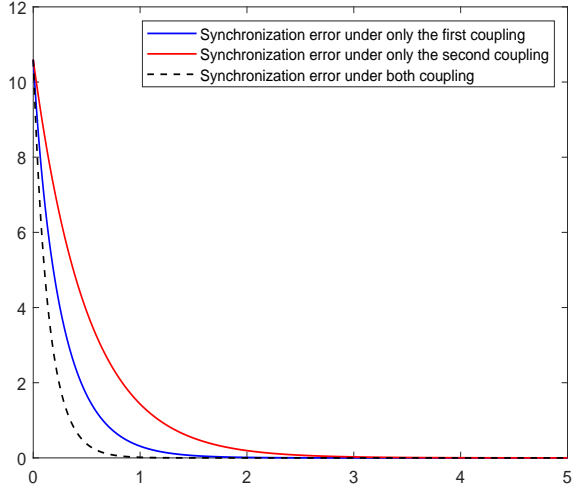}
\caption{Synchronization dynamics under one or two coupling}
\end{center}
\end{figure}

\subsection{Some discussions}

{\bf Exmaple 3:} Consider the following CNMWs model:
\begin{align}\label{example3}
\dot{z}_i(t)=\sum_{j=1}^3G_{ij}^{1} \Gamma^1z_j(t)+\sum_{j=1}^3G_{ij}^{2}\Gamma^2z_j(t),
\end{align}
where $z_i(t)\in R^2$, $G^1$ and $G^2$ are defined in Example 2.

A question arises naturally: \emph{can the above model be rewritten into a CNSW?} If it does, then the synchronization problem has already been solved.

Let $\Gamma^1=\mathrm{diag}([1,2])$ and $\Gamma^2=I_2$, then
\begin{align*}
G^1\otimes \Gamma^1+G^2\otimes \Gamma^2=\left[
\begin{array}{cccccc}
-5&0&2&0&3&0\\
0&-8&0&3&0&5\\
3&0&-6&0&3&0\\
0&3&0&-6&0&3\\
2&0&2&0&-4&0\\
0&3&0&3&0&-6
\end{array}\right],
\end{align*}
which cannot be written in the form $G\otimes \Gamma$. Therefore, the exploration of multi-weighted network model is necessary.

On the other hand, if $\Gamma^1=\Gamma^2=I_2$, for the model (\ref{example3}), it can be rewritten as:
$\dot{z}_i(t)=\sum_{j=1}^3G_{ij}z_j(t)$ with
\begin{align*}
&G=(G_{ij})=\left[
\begin{array}{ccc}
-5& 2& 3\\
3& -6& 3\\
2& 2& -4
\end{array}
\right].
\end{align*}

For matrix $G^{1}$, its NLEVec $\xi^{1}$ is $(0.3,0.2,0.5)^T$; for matrix $G^{2}$, its NLEVec $\xi^{2}$ is $(1/3,1/3,1/3)^T$; and for matrix $G$, its NLEVec $\xi$ is $(0.3214, 0.25, 0.4286)^T$, which can be regarded as a weighted combination of $\xi^{1}$ and $\xi^{2}$, i.e., $\xi_{i}\in [\min(\xi^1_i,\xi^2_i), \max(\xi^1_i,\xi^2_i)], i=1,2,3$. Inspired by these discussions, we will use weighted combination of NLEVec as the reference vector for the study of synchronization.


\begin{myrem}
Since the analysis of networks with multi-weights is similar to that with two-weights, so in the following, we will focus on networks with two coupling matrices.
\end{myrem}

\subsection{Synchronization for a network with two weighted matrices}
We consider a general model with two-weighted matrices,
\begin{align}\label{two}
\dot{z}_i(t)=h(z_i(t))+c\sum_{m=1}^2\sum_{j=1}^NG_{ij}^{m}\Gamma^{m} z_j(t),
\end{align}
where $G^m=(G_{ij}^m), m=1,2$ are both strongly connected and Metzler matrices with zero-row-sum, $\Gamma^m=\mathrm{diag}(\gamma^m_1,\cdots,\gamma^m_n)$ with $\gamma^m_k>0$, $\forall m,k$, and the other variables are all the same with that in model (\ref{single}).

\begin{myrem}
In the model (\ref{two}), the coupling strength for different coupling matrices can be different, such as: $c_1\sum_{j=1}^NG_{ij}^{1}\Gamma^{1} z_j(t)+c_2\sum_{j=1}^NG_{ij}^{2}\Gamma^{2} z_j(t)$, but by letting $c=c_1$ and replacing $G_{ij}^2$ by $c_2G_{ij}^2/c_1$, it will become (\ref{two}).
\end{myrem}

Let $\xi^1$ and $\xi^2$ be the corresponding NLEVec for matrices $G^1$ and $G^2$ in (\ref{two}), respectively. Define a new vector as
\begin{align}\label{combine}
\theta=\mu^1\xi^1+\mu^2\xi^2
\end{align}
where $\mu^1$ and $\mu^2$ are non-negative scalars with $\mu^1+\mu^2=1$,
which are to be determined later.

With this new vector, we can define matrices
\begin{align}
&G_{\theta}^1=[(\Theta-\theta\theta^{T})G^1+(G^1)^{T}(\Theta-\theta\theta^{T})]/2\label{G1theta}\\
&G_{\theta}^2=[(\Theta-\theta\theta^{T})G^2+(G^2)^{T}(\Theta-\theta\theta^{T})]/2\label{G2theta}
\end{align}

\begin{mythm}\label{perf}
For network (\ref{two}), if
\begin{align}\label{absolute}
\max_{i}|\xi^1_i-\xi^2_i|\le \mathcal{ADSB}(G^1)+\mathcal{ADSB}(G^2)
\end{align}
then we can choose scalars $\mu^1$ and $\mu^2$, such that
\begin{align}
0\le\mu^1\le \mathcal{ADSB}(G^2)/\max_{i}|\xi^1_i-\xi^2_i|\label{mu1}\\
0\le\mu^2\le \mathcal{ADSB}(G^1)/\max_{i}|\xi^1_i-\xi^2_i|\label{mu2}
\end{align}
and $\mu^1+\mu^2=1$, therefore, synchronization is realized with
\begin{align}\label{condition2}
L_h+c\sum_{m=1}^2\lambda_2(G^m_{\theta})\min_{k}\gamma_k^m/\|\Theta-\theta\theta^T\|_2<0.
\end{align}
\end{mythm}

\begin{proof} 
For this vector $\theta$ in (\ref{combine}), it is also positive and normalized, so we can use this vector to define the Lyapunov function $V(t)$ in (\ref{lya}), and with the same proof process in Theorem \ref{generalize},
\begin{align}\label{add}
\dot{V}(t)\le 2L_hV(t)+cZ(t)^T[G_{\theta}^1\otimes \Gamma^1+G_{\theta}^2\otimes \Gamma^2]Z(t)
\end{align}
where $G^1_{\theta}$ and $G^2_{\theta}$ are defined in (\ref{G1theta}) and (\ref{G2theta}).

Then, from condition (\ref{estimate}), if inequalities
\begin{align*}
\max_{i}|\theta_i-\xi^1_i|\le \mathcal{ADSB}(G^1),
\max_{i}|\theta_i-\xi^2_i|\le \mathcal{ADSB}(G^2),
\end{align*}
hold, then $G^1_{\theta}$ and $G^2_{\theta}$ would be negative definite in $\mathcal{TS}$.

According to the definition of $\theta$ in (\ref{combine}), the above inequalities can hold if
\begin{align*}
&\mu^2\max_{i}|\xi^1_i-\xi^2_i|\le \mathcal{ADSB}(G^1), \\
&\mu^1\max_{i}|\xi^1_i-\xi^2_i|\le \mathcal{ADSB}(G^2),
\end{align*}
which are equivalent to conditions (\ref{mu1}) and (\ref{mu2}).

We continue the proof from (\ref{add}), if condition (\ref{condition2}) holds, then Expo-Syn can be realized.
\end{proof}

From the viewpoint of addition for vectors, the vector $\theta$ in (\ref{combine}) lies between $\xi^1$ and $\xi^2$, so

If $\xi^1$ and $\xi^2$ are the same, then $\theta=\xi^1=\xi^2$, that is to say, if the two coupling matrices have the same NLEVec, we can naturally use this NLEVec for investigation, which also includes the case that the two matrices are both symmetric.

Otherwise, if $\xi^1$ and $\xi^2$ are different, then,

Case 1: for fixed $\max_{i}|\xi^1_i-\xi^2_i|$, the values $\mu^1$ and $\mu^2$ are determined by indexes $\mathcal{ADSB}(G^m), m=1,2$, whose physical meaning can be described as Fig. 2, obviously, the larger $\mathcal{ADSB}(G^1)$, the larger $\mu^2$, i.e., the vector $\theta$ can deviate from $\xi^1$ larger, and vice versa.
\begin{figure}[htp]
\begin{center}
\includegraphics[width=0.5\textwidth,height=0.25\textheight]{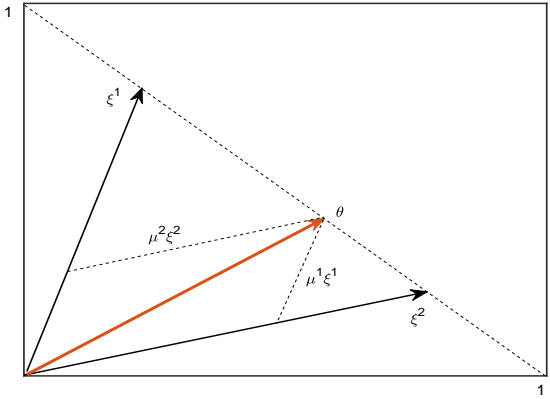}
\caption{Sketch of vector $\theta$ as the weighted addition of vectors $\xi^1$ and $\xi^2$}
\end{center}
\end{figure}

Case 2: for fixed matrices $G^1$ and $G^2$, the smaller $\max_{i}|\xi^1_i-\xi^2_i|$, the more scope for values of $\mu^1$ and $\mu^2$, and several different cases are presented in Fig. 3, where $\omega^{m}=\mathcal{ADSB}(G^{3-m})/\max_{i}|\xi^1_i-\xi^2_i|$, so $\mu^m\le\omega^m, m=1,2$. Here we just list representative cases. Of course, there are other cases which are in some sub-figures of Fig. 3, for example, when $\omega^{1}>1, \omega^{2}<1$, it is similar with Fig. 3 (b).
\begin{figure}[htp]
\begin{center}
\includegraphics[width=0.5\textwidth,height=0.25\textheight]{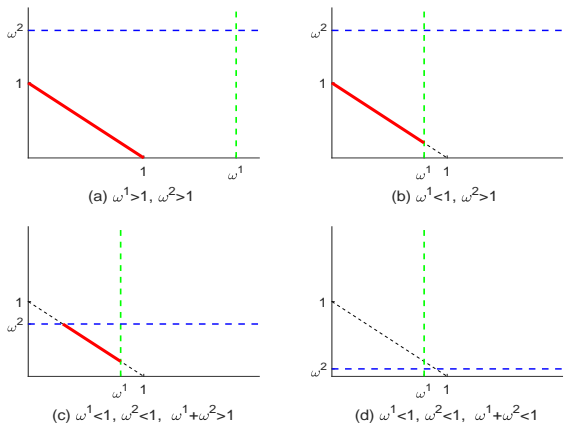}
\caption{The feasible region of $\mu^1$ (horizontal coordinate) and $\mu^2$ (vertical coordinate) under different $\omega^{1}$ and $\omega^{2}$, which are marked in red lines. }
\end{center}
\end{figure}

\begin{myrem}
A natural method for synchronization of CNMWs is to use one NLEVec as the reference vector, for example, $\xi^1$, then for the second matrix, the negative definite analysis of $G^2_{\xi^1}$ in $\mathcal{TS}$ would be the same with Section \ref{generalization}, i.e., $\max_{i}|\xi^1_i-\xi^2_i|\le \mathcal{ADSB}(G^2)$ would be the sufficient condition, and on the other hand, if we use $\xi^2$ as the reference vector, $\max_{i}|\xi^1_i-\xi^2_i|\le \mathcal{ADSB}(G^1)$ would be the sufficient condition. Without loss of generality, we assume that $\mathcal{ADSB}(G^1)<\mathcal{ADSB}({G^2})$. If $\max_{i}|\xi^1_i-\xi^2_i|\le \mathcal{ADSB}(G^1)$, the reference vector can be chosen as $\xi^1$ or $\xi^2$; if $\mathcal{ADSB}(G^1)<\max_{i}|\xi^1_i-\xi^2_i|\le \mathcal{ADSB}({G^2})$, the reference vector can be $\xi^1$, but if $\mathcal{ADSB}({G^2})<\max_{i}|\xi^1_i-\xi^2_i|$, this method would fail. With our method, we can still prove that the two matrices both promote synchronization if
\begin{align*}
\mathcal{ADSB}(G^2)<\max_{i}|\xi^1_i-\xi^2_i|\le \mathcal{ADSB}(G^1)+\mathcal{ADSB}({G^2}).
\end{align*}
\end{myrem}

A {conjecture} for CNMWs is proposed as an open problem:

``\emph{For CNMWs, only if each coupling matrix is strongly connected, then it will synchronize faster than CNSWs.}''\\
That is to say, condition (\ref{absolute}) is not needed at all.

Next, we apply the central adaptive technique on the coupling strength $c$ to realize synchronization.
\begin{mythm}\label{central}
For the network
\begin{align*}
\dot{z}_i(t)=h(z_i(t))+c(t)\sum_{m=1}^2\sum_{j=1}^NG_{ij}^{m}\Gamma^{m} z_j(t),
\end{align*}
If condition (\ref{absolute}) holds, then there exist scalars $\mu^1\ge 0$, $\mu^2\ge 0$ and $\mu^1+\mu^2=1$, such that (\ref{mu1}) and (\ref{mu2}) hold, and we can obtain the reference vector $\theta$. Therefore, synchronization can be finally realized with the adaptive rule
\begin{align*}
\dot{c}(t)=\frac{\beta}{2}\sum_{i=1}^N\theta_i[z_i(t)-\overline{z}(t)]^T[z_i(t)-\overline{z}(t)], ~~\beta>0.
\end{align*}
\end{mythm}

\begin{proof} 
Using the vector $\theta$ in (\ref{combine}), we define
\begin{align*}
V_{a}(t)=V(t)+\frac{\alpha}{\beta}(c^{\star}-c(t))^2,
\end{align*}
where $V(t)$ is defined in (\ref{lya}), $c^{\star}>0$ and $\alpha>0$ satisfy conditions: $L_h-\alpha c^{\star}<0$ and
$g(\alpha)=\sum_{m=1}^2\lambda_2(G^m_{\theta})\min_{k}\gamma_k^m+\alpha\lambda_{\max}(\Theta-\theta\theta^T)<0$, where $\lambda_{\max}(\cdot)$ is the largest EVal.

Then, for $Z(t)\in \mathcal{TS}$, the derivative of $V(t)$ would be
\begin{align*}
\dot{V}_{a}(t)\le &2L_hV(t)+c(t)Z(t)^T[G_{\theta}^1\otimes \Gamma^1+G_{\theta}^2\otimes \Gamma^2]Z(t)\\
&-\alpha(c^{\star}-c(t))Z(t)^T[(\Theta-\theta\theta^T)\otimes I_n]Z(t)\\
\le &2(L_h-\alpha c^{\star})V(t)+c(t)g(\alpha)Z(t)^TZ(t)\le 0
\end{align*}

The rest are similar with that in \cite{liu2010}, here we omit it.
\end{proof}
%

\subsection{Control for a network with two weighted matrices}
Next, we consider the pinning control problem. The network with pinning control added on the first node is
\begin{align}\label{control2}
\dot{z}_i(t)=\left\{
\begin{array}{ll}
h(z_1(t))+c\sum_{m=1}^2\sum_{j=1}^NG_{1j}^{m}\Gamma^{m} z_j(t)\\
\hfill -c\sum_{m=1}^2d^m\Gamma^m(z_1(t)-z(t)),\\
h(z_i(t))+c\sum_{m=1}^2\sum_{j=1}^NG_{ij}^{m}\Gamma^{m} z_j(t),\hfill i\ne 1,
\end{array}
\right.
\end{align}
where $z(t)$ is the synchronization target with $\dot{z}(t)=h(z(t))$, and $d^1>0, d^2>0$.

Define new matrices $\tilde{G}^m=(\tilde{G}_{ij}^m), m=1,2$, where
\begin{align*}
\tilde{G}_{ij}^m=\left\{
\begin{array}{ll}
G_{11}^m-d^m,&i=j=1\\
G_{ij}^m,&\mathrm{otherwise}
\end{array}
\right.
\end{align*}

\begin{mythm}
For network (\ref{control2}), let $\xi^m$ is NLEVec of $G^m, m=1,2$, if
\begin{align}\label{absolute2}
\max_{i}|\xi^1_i-\xi^2_i|\le \mathcal{ADCB}(G^1)+\mathcal{ADCB}(G^2)
\end{align}
then we can choose scalars $\nu^1$ and $\nu^2$, such that
\begin{align*}
0\le\nu^1\le \mathcal{ADCB}(\tilde{G}^2)/\max_{i}|\xi^1_i-\xi^2_i|\\
0\le\nu^2\le \mathcal{ADCB}(\tilde{G}^1)/\max_{i}|\xi^1_i-\xi^2_i|
\end{align*}
and $\nu^1+\nu^2=1$. Therefore, Expo-Syn can be realized if
\begin{align*}
L_h+c\sum_{m=1}^2\lambda_{\max}(_{\theta}\tilde{G}^m)\min_{k}(\gamma_k^m)/\max_{i}{\theta_i}<0,
\end{align*}
where
\begin{align}\label{nucombine}
\theta=\nu^1\xi^1+\nu^2\xi^2.
\end{align}
\end{mythm}

\begin{proof} Since (\ref{absolute2}) holds, so there exists a vector $\theta$ defined in (\ref{nucombine}), such that $\max_{i}|\theta_i-\xi^m_i|\le \mathcal{ADCB}(\tilde{G}^m), m=1,2$, i.e., matrices $(_{\theta}\tilde{G}^m)$ are both negative definite.

Denote $\tilde{z}_i(t)=z_i(t)-z(t)$, and choose the Lyapunov function in (\ref{lya-control}), then
\begin{align*}
\dot{W}(t)\le& 2L_hW(t)+c\sum_{m=1}^2\sum_{i,j=1}^N\theta_i\tilde{z}_i(t)^T\tilde{G}_{ij}^{m}\Gamma^{m} \tilde{z}_j(t)\\
\le&2\bigg[L_h+c\sum_{m=1}^2\lambda_{\max}(_{\theta}\tilde{G}^m)\min_{k}(\gamma_k^m)/\max_{i}{\theta_i}\bigg]W(t)\\
\le&0
\end{align*}
The proof is finished.
\end{proof}

The discussions about the chosen of parameters $\nu^1$ and $\nu^2$ are similar to that for $\mu^1$ and $\mu^2$, here we omit it. Moreover, we can also get the corresponding adaptive rule for coupling strength, we just list the result.

\begin{mythm}
For network (\ref{control2}), suppose the condition (\ref{absolute2}) holds, then we can obtain the reference vector $\theta$ defined in (\ref{nucombine}). Therefore, synchronization to target $z(t)$ can be finally realized with the adaptive rule
\begin{align*}
\dot{c}(t)=\frac{\beta}{2}\sum_{i=1}^N\theta_i[z_i(t)-{z}(t)]^T[z_i(t)-{z}(t)], ~~\beta>0.
\end{align*}
\end{mythm}

\section{Conclusion}\label{con}

The whole paper is carried on around NLEVec of coupling matrices. We firstly generalize the synchronization analysis technique from using NLEVec to any vector, whose Chebyshev distance to NLEVec should be less than $\mathcal{ADSB}$, so this generalization makes the design of Lyapunov function more flexibly than previous works. Then based on this new technique, for directed CNMWs, we can choose another weight coefficients to combine all NLEVec of multiple coupling weights, and prove its validity for (adaptive) synchronization and control. We finally conclude that directed CNMWs can accelerate synchronization than that with a single weight.

This paper just uncovers a corner of CNMWs, and there are still many important questions to be solved, for example, each coupling matrix can be not necessarily strongly connected, and maybe jointly connected condition is enough; non-diagonal elements in coupling matrices can be negative, i.e., the relationship between nodes can be competitive; multiple time delays and distributed adaptive rules can be considered in CNMWs; etc.

\end{document}